\documentclass[conference,letterpaper]{IEEEtran}
\usepackage[utf8]{inputenc}
\usepackage[T1]{fontenc}
\usepackage{url}
\usepackage{cite}
\usepackage{longtable}
\usepackage[cmex10]{amsmath}
\usepackage{algorithm}
\usepackage{algorithmic}
\usepackage{amssymb,amsfonts}
\usepackage{graphicx}
\usepackage{textcomp}
\usepackage{xcolor}
\usepackage{tikz}
\usepackage{ifthen}
\usepackage{array}
\usepackage{hyperref}

\newcommand{\bolds}{\mathbf{s}}

\newcommand{\cA}{\mathcal{A}}
\newcommand{\cB}{\mathcal{B}}

\newcommand{\cK}{\mathcal{K}}

\newcommand{\boldA}{\mathbf{A}}

\newcommand{\boldC}{\mathbf{C}}

\newcommand{\boldS}{\mathbf{S}}

\newcommand{\boldU}{\mathbf{U}}
\newcommand{\boldV}{\mathbf{V}}

\newcommand{\boldY}{\mathbf{Y}}

\newtheorem{pid axiom}{PID Axiom}
\newtheorem{sid axiom}{SID Axiom}

\newtheorem{remark}{Remark}
\newtheorem{definition}{Definition}

\newtheorem{corollary}{Corollary}

\newtheorem{theorem}{Theorem}

\newtheorem{proposition}{Proposition}
\usepackage{ifthen}
\newboolean{showAppendix}
\setboolean{showAppendix}{true}  

\def\BibTeX{{\rm B\kern-.05em{\sc i\kern-.025em b}\kern-.08em
T\kern-.1667em\lower.7ex\hbox{E}\kern-.125emX}}

\newcommand{\E}{\mathbb{E}}
\newcommand{\Cov}{\operatorname{Cov}}

\newcommand{\dd}{\mathrm{d}}
\newcommand{\R}{\mathbb{R}}

\title{Closed-Form Gaussian Estimators for Multi-Source Partial Information Decomposition}

\author{\textbf{Aobo Lyu}$^\star$, \textbf{Andrew Clark}$^{\star}$, and \textbf{Netanel Raviv}$^\dagger$\\
  $^\star$Department of Electrical and Systems Engineering, Washington University in St. Louis, St. Louis, MO, USA\\
  $^\dagger$Department of Computer Science and Engineering, Washington University in St. Louis, St. Louis, MO, USA\\
\texttt{aobo.lyu@wustl.edu}, \texttt{andrewclark@wustl.edu}, \texttt{netanel.raviv@wustl.edu}}

\begin{document}

\maketitle

\begin{abstract}
\ifthenelse{\boolean{showAppendix}}
{}{THIS PAPER IS ELIGIBLE FOR THE STUDENT PAPER AWARD.

}
Computing multi-source partial information decomposition (PID) for continuous data is hard: existing closed-form Gaussian estimators are restricted to two source variables, while continuous arbitrary-source estimators are typically learning-based and do not provide closed-form expressions.
To address this, we develop closed-form Gaussian estimators for multi-source PID.
We provide two-source redundancy, multi-source unique information, the $K$-th order synergistic effect from source subsets of size~$K$, and the total synergistic effect.
The estimators are derived from the conditional-independence-based information measures introduced in our earlier work, under which every quantity reduces to a log-determinant expression in covariance blocks of the system.
The resulting estimator is plug-in consistent, affine invariant, source-permutation symmetric, and additive over independent systems.
We validate it on a controlled Gaussian benchmark, evaluate its computational efficiency against baselines, and confirm its numerical stability in finite-sample regimes.
To our knowledge, this is the first covariance-based closed-form estimator that provides multi-source continuous PID measures.
\end{abstract}

\section{Introduction}
\label{sec:introduction}
Partial information decomposition (PID), introduced by Williams and Beer \cite{williams2010nonnegative}, refines the joint mutual information $I(S_1,\ldots,S_N;T)$ between source variables $S_1,\ldots,S_N$ and a target~$T$ into information atoms, such as redundant, unique, and synergistic information.
Such decompositions are central whenever one needs to know not only \emph{whether} sources predict $T$ but \emph{how} predictive information distributes across them, with applications in neural population coding \cite{schneidman2003synergy}, brain functional connectivity \cite{luppi2022synergistic,gatica2021high}, fairness and feature attribution \cite{liang2023quantifying,dutta2020information}, and dynamical systems \cite{mediano2025toward}.

Despite a decade of progress, no PID estimator for continuous data is simultaneously closed-form and applicable to arbitrary $N \geq 2$.
Existing closed-form Gaussian PID estimators \cite{barrett2015exploration,faes2017multiscale,kay2018exact,niu2019measure,kay2024partial,venkatesh2022partial,venkatesh2023gaussian,gurushankar2022extracting} are covariance-based and sample efficient, but restricted to two sources.
Multivariate measures applicable for arbitrary $N$, such as O-information \cite{rosas2019quantifying} and integrated information decomposition ($\Phi$ID) \cite{mediano2025toward}, admit Gaussian closed forms but produce only a single scalar or operate in a different, non-target-directed framework.
Continuous PID estimators that admit arbitrary $N$, such as the shared-information measures of \cite{schick2021partial,ehrlich2024partial}, rely on nonparametric neighborhood or kernel machinery, and hence do not have closed-form expressions.
Other continuous (non-Gaussian) estimators are primarily two-source~\cite{pakman2021estimating,bara2025partial}.
Restricting to Gaussian variables turns information measures into covariance functionals, enabling stable finite-sample estimation that current methods do not~achieve.

This paper closes this gap by leveraging the conditional-independence-based measures that we introduced in our prior work \cite{lyu2026multivariate} (see Definition~\ref{def:conditional_copy_family}).
This framework is the unique discrete closed-form PID measure simultaneously satisfying Monotonicity, Additivity, Continuity, and Strong Independent Identity~\cite{lyu2026multivariate,lyu2024explicit}, providing principled grounds for our Gaussian extension.
We extend these measures to continuous jointly Gaussian systems, where redundant, unique, and synergistic information admit closed-form log-determinant expressions.
Specifically, we measure the synergistic information that requires source subsets of exactly size~$K \!\le \!N$, which we call the $K\!$-th order \textit{synergistic effect} $SE_K$, and further define a meaningful extension: The \textit{total synergistic effect}~$TSE=\sum_K SE_K$.
These measures can be computed via a single covariance-only identity and enable a closed-form computation that has been missing from the literature.

The contributions of this paper are: (i) a single closed-form expression for the conditional covariance induced by an arbitrary source-subset family (Theorem~\ref{thm:main_identity}), from which the quantities $SE_K$ and~$TSE$ can be computed, as well as unique and redundant information; 
(ii) an efficient computation of TSE without summing over all~$SE_K$, hence applicable to very large-scale systems;
(iii) a combined statement of plug-in consistency, affine invariance, source-permutation symmetry, and additivity over independent systems (Theorem~\ref{thm:structural}), together with a ridge-regularized form (Proposition~\ref{prop:ridge}) for degenerate cases; 
(iv) a parametric Gaussian benchmark with explicit per-subset contributions, on which we verify the computation of~$\{SE_K\}_{K=2}^N$ (which we call the \textit{synergy spectrum}) against discretized PID and non-PID scalar baselines.
The covariance-only structure makes the proposed estimator an off-the-shelf method for downstream tasks---representation learning, fairness auditing, emergence quantification in dynamical systems---without kernel-density or variational machinery.

\section{Preliminaries}
\label{sec:preliminaries}

All logarithms are natural and all information quantities are in nats.
We write $H$ for discrete entropy and $h$ for differential (continuous) entropy.
We use $d_{(\cdot)}$ for the dimension of the variable indicated by the subscript (e.g., $d_T = \dim(T)$, $d_i = \dim(S_i)$).
All covariance and conditional covariance matrices are assumed to be positive definite.

\subsection{Gaussian Information from Covariance}
\label{subsec:gaussian_entropy_review}
For a continuous random vector $X \in \R^{d_X}$ with density $p_X$, the differential entropy is $h(X) \triangleq -\int p_X(x) \log p_X(x)\,\dd x$. For a Gaussian $X \sim \mathcal{N}(\mu, \Sigma)$ with $\Sigma \succ 0$, 
\begin{align}
  h(X)
  &= \frac{1}{2}\log\big((2\pi e)^{d_X} \det\Sigma\big).
  \label{eq:gaussian_entropy}
\end{align}
For jointly Gaussian $(X, Y)$ with $\Sigma_{YY} \succ 0$, the
conditional covariance of $X $ given $ Y$ is $\Sigma_{X \mid Y}
  \triangleq \Sigma_{XX} - \Sigma_{XY}\Sigma_{YY}^{-1}\Sigma_{YX}$,
and the conditional entropy
$h(X \mid Y) = \tfrac{1}{2}\log \big(\!(2\pi e)^{d_X} \!\det\Sigma_{X\mid Y}\!\big)\!$ \cite{cover1999elements}, yielding the mutual information
\begin{align}
  I(X; Y)
  &= \frac{1}{2}\log
  \frac{\det\Sigma_{XX}}{\det\Sigma_{X\mid Y}}.
  \label{eq:gaussian_mi_logdet}
\end{align}
Equation~\eqref{eq:gaussian_mi_logdet} is the template form of all estimators derived in Section~\ref{sec:gaussian_estimators}. Each PID quantity will be expressed as a log-determinant ratio of two covariance matrices.

\subsection{Conditional-Independence-based Information Measures}
\label{subsec:conditional_copy}
We summarize the PID measures of~\cite{lyu2026multivariate} as follows.
Let $\boldS_{[N]} = (S_1, \ldots, S_N)$ denote the sources, with $[N] = \{1, \ldots, N\}$.
For a nonempty subset $\cK \subseteq [N]$ write $\boldS_{\cK} = (S_i)_{i \in \cK}$, and for $K \in [N]$ let $\binom{[N]}{K} = \{\cK \subseteq [N] : |\cK| = K\}$.
The shorthand $\boldS_{\setminus i}$ stands for $\boldS_{[N] \setminus \{i\}}$.

\begin{definition}[Conditional-independent family]
  \label{def:conditional_copy_family}
Let $\boldA$ be a nonempty collection of nonempty subsets of $[N]$.
The conditional-independent family $\{\boldS_{\cK}' : \cK \in \boldA\}$ is defined by
  \begin{align*}
    \Pr\!\left(\bigcap_{\cK \in \boldA}
    \{\boldS_{\cK}' = \bolds_{\cK}\}\;\Big|\; T = t\right)
    \triangleq
    \prod_{\cK \in \boldA}
    \Pr(\boldS_{\cK} = \bolds_{\cK} \mid T = t).
  \end{align*}

Each auxiliary block $\boldS_{\cK}'$ is a copy of $\boldS_{\cK}$ in the sense that it shares the same conditional distribution given $T$, while distinct auxiliary blocks are conditionally independent given $T$.
We write $\boldY_{\boldA} = (\boldS_{\cA_1}', \ldots, \boldS_{\cA_m}')$ for the stacked auxiliary vector when $\boldA = \{\cA_1, \ldots, \cA_m\}$.
When two subsets in $\boldA$ overlap, shared variables appear as independent copies, so $\dim(\boldY_{\boldA}) = \sum_{\cA \in \boldA} d_{\cA}$, where $d_{\cA} = \sum_{i \in \cA} d_i$.
\end{definition}

The conditional-independent construction preserves the marginal conditional probability $\Pr(\boldS_{\cK} \mid T)$ for every $\cK \in \boldA$ but breaks the high-order dependencies among the subsets given $T$. 
Intuitively, it isolates the part of the joint structure that is determined by source-subset-to-target marginals from the part that comes from how the subsets jointly organize around $T$.
It thereby exposes distinct PID components when classical information-theoretic quantities are applied to $(\boldY_{\boldA}, T)$ versus $(\boldS_{[N]}, T)$, as formalized in Definition~\ref{def:measures_quantities}.

\begin{definition}[Information Measures]
  \label{def:measures_quantities}
Given target $T$ and sources $\boldS_{[N]}$, the PID measures in~\cite{lyu2026multivariate} are defined as follows.
 
For two sources $S_1$ and $S_2$, the redundancy is defined as
{
\setlength{\abovedisplayskip}{5pt}
\setlength{\belowdisplayskip}{5pt}
  \begin{align}
    \operatorname{Red}(S_1, S_2 \to T)
    &\triangleq I(S_1'; S_2').
    \label{eq:def_red_twosource}
  \end{align}
}
 
For multiple sources $\boldS_{[N]}$, the general unique information corresponds to the family $\boldA = \{\{i\}, [N] \setminus \{i\}\}$, with auxiliary vector $\boldY_{\boldA} = (S_i', \boldS'_{\setminus i})$.
It is defined as
  \begin{align}
    \operatorname{Un}(S_i \to T \mid \boldS_{\setminus i})
    &\triangleq I(S_i'; T \mid \boldS'_{\setminus i}).
    \label{eq:def_unique_multi}
  \end{align}
The $K$-th order synergistic effect, for $K = 2, \ldots, N$, with $\boldC_K \triangleq \binom{[N]}{K}$, is
  \begin{align}
    \operatorname{SE}_K
    &\triangleq H(T \mid \boldY_{\boldC_{K-1}})
    - H(T \mid \boldY_{\boldC_K}).
    \label{eq:def_se_k}
  \end{align}
The narrow synergy --- so named because it captures synergistic information available only from all $N$ sources jointly --- and total synergistic effect are
  \begin{align}
    &\!\operatorname{Syn}(\boldS_{[N]} \!\to\! T)
    \triangleq \operatorname{SE}_N=H(T \!\mid\! \boldY_{\boldC_{N\!-\!1}}\!)
    - H(T \mid \boldS_{[N]}), \text{ and} \nonumber\\
    &\!\operatorname{TSE}(\boldS_{[N]} \to T)
    \triangleq \sum_{K=2}^N \operatorname{SE}_K=H(T \mid \boldY_{\boldC_{1}})
    - H(T \mid \boldS_{[N]}).\nonumber
  \end{align}

For $N \geq 3$, we deliberately do not assign a redundancy. See \cite{lyu2026structural} for the structural reason. Nevertheless, the narrow synergy~$\operatorname{Syn}$ provides a building block from which synergy-based PID frameworks~\cite{ince2017measuring} can reconstruct all PID atoms.
\end{definition}

The synergy spectrum $(\operatorname{SE}_2, \ldots, \operatorname{SE}_N)$ distributes synergy across interaction orders, where each $\operatorname{SE}_K$ captures the information about $T$ that requires source subsets of exactly size $K$ to act jointly, with narrow synergy as the highest-order component and TSE as the all-orders aggregate.
The index $K$ thus denotes both the synergistic-effect order and the cardinality of the source subsets in $\boldC_K$.
 
These quantities are the estimands of our Gaussian estimators in the next section.
The central technical task is to express $h(T \mid \boldY_{\boldA})$ in closed form when $(T, \boldS_{[N]})$ are jointly Gaussian. Section~\ref{sec:gaussian_estimators} carries this out via a single covariance identity for $\Psi_{\boldA} = \Cov(T \mid \boldY_{\boldA})$, from which all four families of quantities above follow as corollaries.

\section{Gaussian Estimators for Multi-Source Information Hierarchies}
\label{sec:gaussian_estimators}
This section presents our main technical results.
We state the main identity (Section~\ref{subsec:main_identity}) and specialize to closed-form measures and complexity analysis (Section~\ref{subsec:corollaries}).

\subsection{The Main Covariance Identity}
\label{subsec:main_identity}

Let $(T, S_1, \ldots, S_N)$ be jointly Gaussian with positive-definite joint covariance, and write $\Sigma_T, \Sigma_{\cA_1 \cA_2}, \Sigma_{\cA T}$ for covariance and cross-covariance blocks indexed by $\cA_1, \cA_2 \subseteq [N]$.
For a family $\boldA = \{\cA_1, \ldots, \cA_m\}$ of source subsets, the conditional-independent construction (Definition~\ref{def:conditional_copy_family}) produces the auxiliary vector $\boldY_{\boldA} = (\boldS_{\cA_1}', \ldots, \boldS_{\cA_m}') \in \R^{D_{\boldA}}$, $D_{\boldA} = \sum_{\cA \in \boldA} d_{\cA}$, where $\boldS_{\cA_1}', \ldots, \boldS_{\cA_m}'$ are conditionally independent given $T$, and each pair $(\boldS_{\cA_a}',T)$ has the same joint distribution as $(\boldS_{\cA_a},T)$.
Let $\Psi_{\boldA} \triangleq \Cov(T \mid \boldY_{\boldA})$.
The next main identity is the computational core of the paper.

\begin{theorem}[Main identity]
  \label{thm:main_identity}
Under the assumptions above, $(T, \boldY_{\boldA})$ is jointly Gaussian, with cross- and joint-covariance 
  \begin{align}
    \Cov(\boldS_{\cA_a}', T)
    &= \Sigma_{\cA_a T}, \text{ for all } a\in[m], \text{ and}
    \label{eq:main_T_block}\\
    \Cov(\boldS_{\cA_a}', \boldS_{\cA_b}')
    &=
    \begin{cases}
      \Sigma_{\cA_a \cA_a}, & a = b,\\
      \Sigma_{\cA_a T}\Sigma_T^{-1}\Sigma_{T \cA_b}, & a \neq b.
    \end{cases}
    \label{eq:main_cross_block}
  \end{align}

  \begin{align}
    \text{Define } \qquad \Lambda_{\boldA}
    &\triangleq
    \begin{bmatrix}
      \Sigma_{\cA_1 T} \\
      \vdots \\
      \Sigma_{\cA_m T}
    \end{bmatrix}
    \in \R^{D_{\boldA} \times d_T},
    \label{eq:Lambda_def}
  \end{align}
  and let $\Gamma_{\boldA} \in \R^{D_{\boldA} \times D_{\boldA}}$ be the block matrix with entries given by~\eqref{eq:main_cross_block}; then the covariance of $(T, \boldY_{\boldA})$ has the block form
\begin{align}
    \Cov\!\begin{pmatrix} T \\ \boldY_{\boldA} \end{pmatrix}
    &= \begin{pmatrix} \Sigma_T & \Lambda_{\boldA}^\top \\ \Lambda_{\boldA} & \Gamma_{\boldA} \end{pmatrix}.
    \label{eq:joint_cov_block}
  \end{align}
    Both $\Gamma_{\boldA} \succ 0$ and $\Psi_{\boldA} \succ 0$, and $\Psi_{\boldA}$ is the Schur complement of $\Gamma_{\boldA}$ in~\eqref{eq:joint_cov_block}, i.e.,
  \begin{align}
    \Psi_{\boldA}
    &= \Sigma_T - \Lambda_{\boldA}^\top \Gamma_{\boldA}^{-1}
    \Lambda_{\boldA}.
    \label{eq:main_Psi}
  \end{align}
\end{theorem}

\begin{IEEEproof}[Proof sketch]
For each $a$, $(T, \boldS_{\cA_a}')$ has the same joint distribution as $(T, \boldS_{\cA_a})$, so each pair is jointly Gaussian. Combined with conditional independence of $\{\boldS_{\cA_a}'\}_a$ given $T$, this yields joint Gaussianity of $(T, \boldY_{\boldA})$.

Equation~\eqref{eq:main_T_block} follows from the same conditional distribution of $(\boldS_{\cA_a}',T)$ and $(\boldS_{\cA_a},T)$.
For~\eqref{eq:main_cross_block}: when $a = b$, $\boldS_{\cA_a}'$ shares the marginal of $\boldS_{\cA_a}$, giving the diagonal block; when $a \neq b$, conditional independence yields $\Cov(\boldS_{\cA_a}', \boldS_{\cA_b}' \mid T) = 0$, so the law of total covariance reduces it to $\Cov(\E[\boldS_{\cA_a} \mid T], \E[\boldS_{\cA_b} \mid T]) = \Sigma_{\cA_a T}\Sigma_T^{-1}\Sigma_{T \cA_b}$.

Equation~\eqref{eq:main_Psi} is the Schur complement of $\Gamma_{\boldA}$ in~\eqref{eq:joint_cov_block}, with positive definiteness of $\Gamma_{\boldA}, \Psi_{\boldA}$ established via the Woodbury identity.
See 
\ifthenelse{\boolean{showAppendix}}
{Appendix~\ref{app:main_identity}}{the full version~\cite{lyu2026closedformfull}}
for a detailed proof.
\end{IEEEproof}

As a special case, when $\boldA = \{\cA\}$ contains a single block, $\Gamma_{\{\cA\}} = \Sigma_{\cA \cA}$, $\Lambda_{\{\cA\}} = \Sigma_{\cA T}$, and $\Psi_{\{\cA\}} = \Sigma_T - \Sigma_{T\cA}\Sigma_{\cA\cA}^{-1}\Sigma_{\cA T} = \Cov(T \mid \boldS_{\cA})$.

\subsection{Closed-Form PID Estimators}
\label{subsec:corollaries}
The corollary below uses the results of Theorem~\ref{thm:main_identity} to write each measure of Definition~\ref{def:measures_quantities} as a log-determinant ratio of covariance matrices.

\begin{corollary}[Closed-form Gaussian PID measures]
  \label{cor:closed_forms}
Under joint Gaussianity, the measures of Definition~\ref{def:measures_quantities} have the following closed forms.
Let $\boldU_i \triangleq \{\{i\}, [N] \setminus \{i\}\}$ and $\boldV_i \triangleq \{[N] \setminus \{i\}\}$, and write $\Sigma_{ii} \triangleq \Cov(S_i)$ for the marginal covariance of $S_i$; the matrices $\Gamma_{\boldU_i}, \Psi_{\boldU_i}, \Psi_{\boldV_i}$ are obtained from Theorem~\ref{thm:main_identity}.

The two-source redundancy ($N = 2$) is
\begin{align}
    \operatorname{Red}(S_1, S_2 \to T)
    &= \tfrac{1}{2}\log
    \tfrac{\det \Sigma_{11} \cdot \det \Sigma_{22}}
    {\det \Gamma_{\boldU_1}}.
    \label{eq:cor_red}
  \end{align}
For $N \geq 2$, the per-source unique information is
  \begin{align}
    \operatorname{Un}(S_i \to T \mid \boldS_{\setminus i})
    &= \tfrac{1}{2}\log
    \tfrac{\det \Psi_{\boldV_i}}
    {\det \Psi_{\boldU_i}},
    \quad i \in [N].
    \label{eq:cor_unique}
  \end{align}
The total synergistic effect counts synergy across all orders,
  \begin{align}
    \operatorname{TSE}(\boldS_{[N]} \to T)
    = \sum_{K=2}^N \operatorname{SE}_K
    &= \tfrac{1}{2}\log
    \tfrac{\det \Psi_{\boldC_1}}
    {\det \Psi_{\boldC_N}},
    \label{eq:cor_tse}
  \end{align} 
where the $K$-th order synergistic effect is
 \begin{align}
    \operatorname{SE}_K
    &= \tfrac{1}{2}\log
    \tfrac{\det \Psi_{\boldC_{K-1}}}
    {\det \Psi_{\boldC_K}},
    \quad K = 2, \ldots, N,
    \label{eq:cor_SE_K}
  \end{align}
and narrow synergy is the $K = N$ specialization,
  \begin{align}
    \operatorname{Syn}(\boldS_{[N]} \to T)
    &= \operatorname{SE}_N = \tfrac{1}{2}\log
    \tfrac{\det \Psi_{\boldC_{N-1}}}
    {\det \Psi_{\boldC_N}}.
    \label{eq:cor_narrow_synergy}
  \end{align}
\end{corollary}

The TSE expression \eqref{eq:cor_tse} arises by telescoping the SE$_K$ ratios over $K = 2, \ldots, N$.
TSE is computable without combinatorial enumeration over $\boldC_K$: it requires only $\Psi_{\boldC_1}$ and $\Psi_{\boldC_N}$, each at the original source dimension $\sum_i d_i$.

\begin{proposition}[Complexity]
  \label{prop:complexity}
Evaluating $\Psi_{\boldA}$ via \eqref{eq:main_Psi} reduces to a Cholesky factorization of $\Gamma_{\boldA}$ at cost $\mathcal{O}(D_{\boldA}^3)$.
Consequently, for scalar sources, the full spectrum costs $\mathcal{O}(\sum_{K=2}^{N} K^3 \binom{N}{K}^3)$ (dominated by $K \approx N/2$), while $\operatorname{TSE}$ \eqref{eq:cor_tse} requires only $D_{\boldC_1} = \sum_i d_i$ at cost $\mathcal{O}((\sum_i d_i)^3)$, and each $\operatorname{Un}(S_i \to T \mid \boldS_{\setminus i})$ \eqref{eq:cor_unique} requires only $D_{\boldU_i} = \sum_i d_i$ at the same cost per source.
\end{proposition}

TSE and per-source unique information are therefore tractable at scales where the full $K$-spectrum is not.

\section{Structural Properties and Regularization}
\label{sec:properties}

This section establishes the structural properties of the proposed measures in a single combined statement (Theorem~\ref{thm:structural}), and provides a ridge-regularized form (Proposition~\ref{prop:ridge}) for degenerate population cases.

\begin{theorem}[Structural properties]
  \label{thm:structural}
Each Gaussian PID measure $\theta$ defined in Section~\ref{sec:gaussian_estimators} satisfies the following.

\noindent (i) \emph{Blockwise affine invariance.} For invertible matrices $A_T, A_1, \!\ldots,\! A_N$ and shifts $b_T, b_1,\! \ldots,\! b_N$, define $\widetilde{T} = A_T T \!+\! b_T$ and $\widetilde{S}_i \!= \!A_i S_i \!+\! b_i$. Then $\theta(\widetilde{T}, \widetilde{S}_1, \ldots, \!\widetilde{S}_N) = \theta(T, S_1, \ldots,\! S_N)$. The transformations $A_i$ need not have unit determinant.

\noindent (ii) \emph{Source-permutation symmetry.} For any permutation $\pi$ of $[N]$ and $\widetilde{S}_i = S_{\pi(i)}$, $\operatorname{SE}_K$, $\operatorname{Syn}$, and $\operatorname{TSE}$ are invariant; the general unique information is equivariant in the sense $\operatorname{Un}(\widetilde{S}_i \to T \mid \widetilde{\boldS}_{\setminus i}) = \operatorname{Un}(S_{\pi(i)} \to T \mid \boldS_{\setminus \pi(i)})$.

\noindent (iii) \emph{Additivity over independent systems.}
 If $(T^{(a)}, S_1^{(a)}, \ldots, S_N^{(a)})$ and $(T^{(b)}, S_1^{(b)}, \ldots, S_N^{(b)})$ are jointly independent Gaussian systems with source blocks indexed by the same set $[N]$, and we set $T = (T^{(a)}, T^{(b)})$ and $S_i = (S_i^{(a)}, S_i^{(b)})$, then $\theta(T, \boldS_{[N]}) = \theta(T^{(a)}, \boldS_{[N]}^{(a)}) + \theta(T^{(b)}, \boldS_{[N]}^{(b)})$.

\noindent (iv) \emph{Plug-in consistency.} If the population covariance $\Sigma$ of $(T, S_1, \ldots, S_N)$ is positive definite and the empirical covariance $\widehat{\Sigma}_M$ from $M$ i.i.d.\ samples satisfies $\widehat{\Sigma}_M \xrightarrow[M \to \infty]{\mathrm{a.s.}} \Sigma$, then almost surely the plug-in estimator $\theta(\widehat{\Sigma}_M)$ converges to $\theta(\Sigma)$ and is well defined for all large $M$.
\end{theorem}

\begin{IEEEproof}[Proof sketch]
(i) Each $\Psi_{\boldA}$ transforms as $\widetilde{\Psi}_{\boldA} = A_T \Psi_{\boldA} A_T^\top$ under blockwise affine change, since invertible per-block transforms do not change $\Cov(T \mid \boldY_{\boldA})$ up to conjugation by $A_T$.
The factor $\log\det(A_T A_T^\top)$ then cancels in every log-determinant ratio.
(ii) The families $\boldC_K$ are defined set-theoretically over $[N]$ and are invariant under relabeling, so $\Psi_{\boldC_K}$ is invariant.
(iii) Independence makes all covariance blocks block-diagonal; $\Psi_{\boldA}$ inherits the block-diagonal structure and $\log\det$ splits additively across the two systems.
(iv) $\theta$ is a composition of continuous matrix operations on the open PD cone, so the continuous mapping theorem yields the claim.
Detailed proofs are in 
\ifthenelse{\boolean{showAppendix}}
{Appendix~\ref{app:structural}.}{the full version~\cite{lyu2026closedformfull}.}
\end{IEEEproof}

For numerical settings where the empirical covariance is near-singular, we use a ridge-regularized form.
Let
\begin{align}
  \widehat{\Sigma}_\lambda
  &\triangleq \widehat{\Sigma} + \lambda I,
  \qquad \lambda > 0,
  \label{eq:ridge_def}
\end{align}
and write $\theta(\widehat{\Sigma}_\lambda)$ for the estimator obtained by substituting $\widehat{\Sigma}_\lambda$ for $\Sigma$ in the formulas of Section~\ref{sec:gaussian_estimators}.

\begin{proposition}[Ridge-regularized estimator]
  \label{prop:ridge}
For every $\lambda > 0$ and every $\widehat{\Sigma} \succeq 0$:
  \begin{itemize}
    \item[(a)] $\theta(\widehat{\Sigma}_\lambda)$ is finite and continuous in $\widehat{\Sigma}$;
    \item[(b)] if $\widehat{\Sigma} \succ 0$ and all Schur complements arising in $\theta$ are positive definite, then $\theta(\widehat{\Sigma}_\lambda) \to \theta(\widehat{\Sigma})$ as $\lambda \to 0$.
  \end{itemize}
\end{proposition}
\begin{IEEEproof}[Proof sketch]
$\widehat{\Sigma}_\lambda \succ 0$ for $\lambda > 0$ ensures all Schur complements are PD, so $\theta(\widehat{\Sigma}_\lambda)$ is finite; continuity and the $\lambda \to 0$ limit follow from continuity of matrix operations on the PD cone. Detailed proofs are in
\ifthenelse{\boolean{showAppendix}}
{Appendix~\ref{app:ridge}.}{the full version~\cite{lyu2026closedformfull}.}
\end{IEEEproof}

The ridge form covers degenerate population cases (e.g., $T$ a deterministic function of $\boldS$); for typical samples, unregularized plug-in suffices, as confirmed empirically in \ifthenelse{\boolean{showAppendix}}
{Appendix~\ref{app:ridge_exp}.}{the full version~\cite{lyu2026closedformfull}.}

\begin{remark}[Sign of $\operatorname{SE}_K$]
  \label{rem:sign}
Since $\boldY_{\boldC_{K-1}}$ and $\boldY_{\boldC_K}$ are not nested, the standard ``more conditioning cannot increase entropy'' argument does not apply, and $\operatorname{SE}_K$ need not be nonnegative.
Operationally, $\operatorname{SE}_K < 0$ means $\boldY_{\boldC_K}$ carries less information about $T$ than $\boldY_{\boldC_{K-1}}$, despite $\boldC_K$ consisting of larger subsets. 
A sufficient condition is the residual ordering $\Psi_{\boldC_{K-1}} \succeq \Psi_{\boldC_K}$, under which monotonicity of $\log\det$ on the positive-definite cone gives $\operatorname{SE}_K \geq 0$.
\end{remark}

\section{Experiments}
\label{sec:experiments}
We validate effectiveness and efficiency of the proposed methods through the following two experiments; additional ridge-behavior, finite-sample convergence, and two-source estimator-comparison experiments are in 
\ifthenelse{\boolean{showAppendix}}
{Appendices~\ref{app:ridge_exp},~\ref{app:convergence}, and~\ref{app:n2_comparison}.}{the full version~\cite{lyu2026closedformfull}.}
Code is publicly available at \url{https://github.com/LvAobo/gaussian-pid}.

\subsection{Synergy Spectrum Recovery and Subset Localization}
\label{subsec:exp_recovery}

\begin{figure*}[t]
  \centering
  \includegraphics[width=0.85\textwidth, trim={3 2 3 2}, clip]{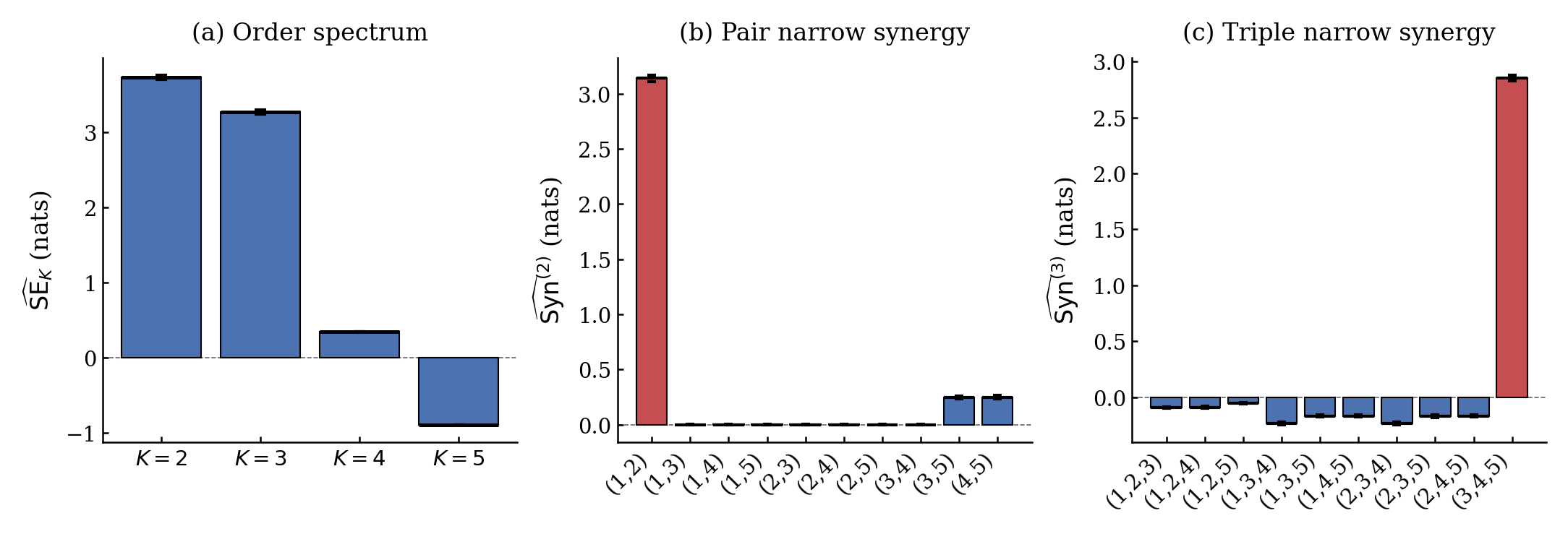}
  \caption{Synergy spectrum and subset narrow synergy on the benchmark in Section~\ref{subsec:exp_recovery}. (a) Global spectrum
    $\widehat{\operatorname{SE}}_K$ for $K = 2, \ldots,
    5$. (b) Two-source
    synergy across all $\binom{5}{2} = 10$ pairs. (c) Three-source
   synergy across all $\binom{5}{3} = 10$ triples. Error
    bars show $\pm 1$ standard deviation across $50$ trials at
  $M = 1000$.}
  \label{fig:exp1}
\end{figure*}

\paragraph{Benchmark}
To verify that the estimator recovers both the dominant interaction order and the source subsets responsible for the synergy, we construct a $5$-source jointly Gaussian system with two independent subsystems contributing synergy at different orders through linear noise cancellation.
Let $T_2, T_3 \sim \mathcal{N}(0, 1)$ be independent variables, $T = (T_2, T_3)$, $U, V_1, V_2 \sim \mathcal{N}(0, 2)$, and i.i.d. $\epsilon_i \sim \mathcal{N}(0, 0.05), i \in [5]$. Define
\begin{align}
  S_1 &= T_2 + U + \epsilon_1, & S_2 &= T_2 - U + \epsilon_2, \notag\\
  S_3 &= T_3 + V_1 + \epsilon_3, & S_4 &= T_3 + V_2 + \epsilon_4, \notag\\
  S_5 &= T_3 - V_1 - V_2 + \epsilon_5.
\end{align}
The designed interaction structure $S_1 + S_2 \perp U$ and $S_3 + S_4 + S_5 \perp V_1, V_2$ produces dominant second- and third-order synergistic effects in the subsets $\{S_1, S_2\}$ and $\{S_3, S_4, S_5\}$.
By Theorem~\ref{thm:structural}(iii), the PID measures split additively across the two subsystems.

\paragraph{Setup}
We compute the plug-in estimator over $50$ independent trials of $M = 1000$ samples; population values are obtained in closed form. Since the order families $\boldC_{K-1}, \boldC_K$ are not nested, $\operatorname{SE}_K$ is a signed effect (Remark~\ref{rem:sign}).
 
\paragraph{Results}
The plug-in mean differs from the population value by less than $2.2 \times 10^{-3}$ for every quantity, with standard deviations under $0.031$ as shown in Figure~\ref{fig:exp1}.
Panel (a) shows the spectrum: $\operatorname{SE}_2 = 3.73$ and $\operatorname{SE}_3 = 3.27$ nats dominate, while $K = 4, 5$ contribute $+0.34$ and $-0.90$ nats; the negative $\operatorname{SE}_5$ illustrates the signed nature of the spectrum (Remark~\ref{rem:sign}). The total is $\operatorname{TSE} = 6.45$ nats.
Panels (b) (c) localize the subsets: the pair $\{S_1, S_2\}$ attains $3.14$ nats ($>10\times$ the next pair) and the triple $\{S_3, S_4, S_5\}$ attains $2.85$ nats (isolated among $10$ triples), confirming both order and subset recovery.

\subsection{Computational Scalability}
\label{subsec:exp_scalability}

\paragraph{Setup}
To verify that the Gaussian estimators for TSE, per-source unique information, and narrow synergy remain tractable at scales where the full $K$-spectrum and discrete baselines become prohibitive, we use scalar sources $d_i = 1$ and $N \in [2, 500]$, drawing $M = 1000$ samples from $\mathcal{N}(0, \Sigma_N)$ where $\Sigma_N = \frac{1}{N} A A^\top + 0.5 I$ and $A$ has i.i.d.\ standard normal entries; the offset bounds the condition number.
We benchmark four Gaussian estimators of this work---full spectrum, $\operatorname{TSE}$, $\operatorname{Un}$, $\operatorname{Syn}$---against discrete baselines (conditional-independent~\cite{lyu2026multivariate}; Iccs~\cite{ince2017measuring}, with each scalar binned into $b=3$ bins giving a $3^{N+1}$ table) and Gaussian O-information~\cite{rosas2019quantifying}.
Although O-information is a scalar redundancy/synergy summary rather than a target-directed PID measure, its Gaussian expression and frequent use in high-dimensional applications make it a natural computational-efficiency baseline; our methods aim to match its scalability while directly measuring synergistic structure.
Wall-clock time is reported as the median over $10$ trials under a $10^3$-second budget per method.
 
\begin{figure*}[t]
  \centering
  \includegraphics[width=0.875\textwidth, trim={3 2 3 2}, clip]{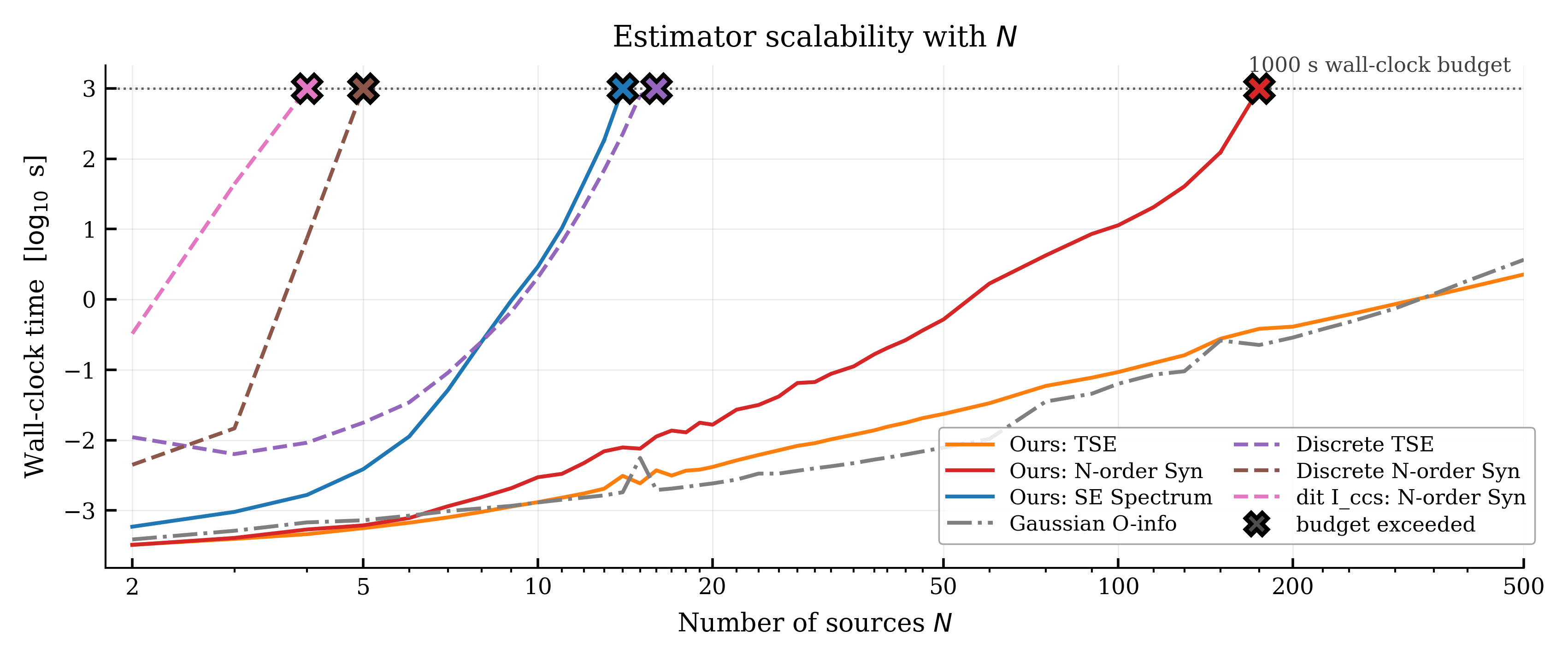}
  \caption{Wall-clock time versus number of sources $N$ on a
  log scale. Solid lines are the Gaussian estimators
  of this work; dashed lines are discrete baselines;
  dash-dot is Gaussian O-information. Median over $10$ trials at $M = 1000$ samples.
  Crosses mark the $N$ at which a method exceeds the
  $10^3$-second budget.}
  \label{fig:exp2}
\end{figure*}
 
\paragraph{Results}
Figure~\ref{fig:exp2} separates the methods into three regimes. Discrete baselines exceed the budget at very small $N$ ($N = 4$ for Iccs, $N = 5$ for discrete narrow synergy, $N = 15$ for discrete TSE). The full Gaussian $K$-spectrum exceeds the budget at $N = 15$, reflecting the combinatorial cost $\mathcal{O}(\sum_K K^3 \binom{N}{K}^3)$. In contrast, $\operatorname{TSE}$, $\operatorname{Un}$, and Gaussian O-information remain under $1$ second at $N = 300$, while $\operatorname{Syn} \in \mathcal{O}(N^6)$ exceeds the budget only near $N = 200$, substantiating Proposition~\ref{prop:complexity}.

\section{Comparison and Conclusion}
\label{sec:comparison_conclusion}
\subsection{Comparison with Existing PID Estimators}
\label{subsec:comparison_existing}
\begin{table}[t]
  \centering
  \caption{Comparison with related continuous PID estimators.}
  \label{tab:comparison_existing_methods}
  \renewcommand{\arraystretch}{1.15}
  \setlength{\tabcolsep}{4pt}
  \footnotesize
  \begin{tabular}{>{\centering\arraybackslash}p{0.355\columnwidth}>{\centering\arraybackslash}p{0.13\columnwidth}>{\centering\arraybackslash}p{0.095\columnwidth}>{\centering\arraybackslash}p{0.1125\columnwidth}>{\centering\arraybackslash}p{0.1525\columnwidth}}
    \hline
    Method & Scope & Closed form & $N=2$ PID & $N\geq 3$ Syn + Spectrum\\
    \hline
    Barrett~\cite{barrett2015exploration}
      & Gaussian & Yes
      & \checkmark & ---\\
    Faes et al.~\cite{faes2017multiscale}
      & Gaussian Process & Yes
      & \checkmark & ---\\
    Kay--Ince 
     \cite{kay2018exact}
      & Gaussian & Yes
      & \checkmark & ---\\
    Niu--Quinn, Kay\cite{niu2019measure,kay2024partial}
      & Gaussian & Almost
      & \checkmark & ---\\
    Venkatesh--Schamberg \cite{venkatesh2022partial}
      & Gaussian & Partial
      & \checkmark & ---\\
    Venkatesh et al. \cite{venkatesh2023gaussian}
      & Gaussian & Plug-in
      & \checkmark & ---\\
    Gurushankar
      et al. \cite{gurushankar2022extracting}
      & Gaussian & Yes
      & Only Un & ---\\
    Pakman et al. \cite{pakman2021estimating}
      &Continuous & Learning
      & \checkmark & ---\\
    Ehrlich
      et al. \cite{ehrlich2024partial}
      &Continuous & Learning
      & \checkmark & ---\\
    O-information \cite{rosas2019quantifying}
      & Gaussian & Yes
      & --- & ---\\
    $\Phi$ID~\cite{mediano2025toward}
      & Gaussian & Yes
      & --- & ---\\
    \textbf{This work}
      & \textbf{Gaussian} & \textbf{Yes}
      & \checkmark & \checkmark\\
    \hline
  \end{tabular}
  \vspace{2pt}

{\scriptsize Almost: analytic + optimization; Partial: subset of atoms; Plug-in: covariance plug-in.\par}
\end{table}

Table~\ref{tab:comparison_existing_methods} positions the proposed estimator within the landscape of Gaussian and continuous PID-related work.
Existing closed-form Gaussian PID estimators are restricted to $N = 2$ source variables \cite{barrett2015exploration,faes2017multiscale,kay2018exact, niu2019measure,kay2024partial,venkatesh2022partial, venkatesh2023gaussian,gurushankar2022extracting}; multivariate Gaussian methods all give a non-target-directed structure ($\Phi$ID) or a scalar (O-information).
Our conditional-independent framework yields a closed-form Gaussian estimator naturally defined for any $N \geq 2$ without optimization or learning, and is the first to provide a synergy spectrum together with an aggregate TSE in closed form.
More broadly, the closed-form multi-source $\operatorname{Syn}$ provides a building block from which all atoms can be reconstructed via lattice inversion~\cite{ince2017measuring,kolchinsky2022novel}, despite known structural inconsistency analyzed in~\cite{lyu2026structural}.

\subsection{Conclusion}
\label{subsec:conclusion}

This paper develops closed-form Gaussian estimators for the conditional-independent PID of \cite{lyu2026multivariate}.
A single covariance identity (Theorem~\ref{thm:main_identity}) yields log-determinant expressions for two-source redundancy, general unique information, total synergistic effect, and $K$-th order synergistic effects across any $N \geq 2$ (Corollary~\ref{cor:closed_forms}).
The estimator is plug-in consistent and satisfies blockwise affine invariance, source-permutation symmetry, and additivity over independent systems (Theorem~\ref{thm:structural}), and admits a regularized form for degenerate cases (Proposition~\ref{prop:ridge}).
Empirically, the estimator recovers the synergy spectrum on a controlled benchmark and remains tractable for $N$ in the hundreds, where discrete plug-in baselines are infeasible.

Two natural extensions remain.
First, the covariance-only structure suggests Gaussian copula transforms~\cite{ince2017statistical} as a route to non-Gaussian continuous data.
Second, a deeper analysis of negative $\operatorname{SE}_K$ (Remark~\ref{rem:sign}), i.e., characterizing when $\Psi_{\boldC_{K-1}} \succeq \Psi_{\boldC_K}$ holds, would clarify the interpretation of $\operatorname{SE}_K$.

\newpage
\IEEEtriggeratref{14}
\bibliographystyle{unsrt}
\bibliography{references}

@article{lyu2026multivariate,
  title={Multivariate partial information decomposition: Constructions, inconsistencies, and alternative measures},
  author={Lyu, Aobo and Clark, Andrew and Raviv, Netanel},
  journal={Physical Review E},
  volume={113},
  number={3},
  pages={034102},
  year={2026},
  publisher={APS}
}

@article{lyu2026structural,
  title={Structural Impossibility of Antichain-Lattice Partial Information Decomposition},
  author={Lyu, Aobo and Clark, Andrew and Raviv, Netanel},
  journal={arXiv preprint arXiv:2604.03869},
  year={2026}
}

@book{cover1999elements,
  title={Elements of information theory},
  author={Cover, Thomas M},
  year={1999},
  publisher={John Wiley \& Sons}
}

@inproceedings{lyu2024explicit,
  title={Explicit formula for partial information decomposition},
  author={Lyu, Aobo and Clark, Andrew and Raviv, Netanel},
  booktitle={2024 IEEE International Symposium on Information Theory (ISIT)},
  pages={2329--2334},
  year={2024},
  organization={IEEE}
}

@article{williams2010nonnegative,
  title={Nonnegative decomposition of multivariate information},
  author={Williams, Paul L and Beer, Randall D},
  journal={arXiv preprint arXiv:1004.2515},
  year={2010}
}

@article{ince2017measuring,
  title={Measuring multivariate redundant information with pointwise common change in surprisal},
  author={Ince, Robin AA},
  journal={Entropy},
  volume={19},
  number={7},
  pages={318},
  year={2017},
  publisher={MDPI}
}

@article{kolchinsky2022novel,
  title={A novel approach to the partial information decomposition},
  author={Kolchinsky, Artemy},
  journal={Entropy},
  volume={24},
  number={3},
  pages={403},
  year={2022},
  publisher={MDPI}
}

@article{barrett2015exploration,
  title={Exploration of synergistic and redundant information sharing in static and dynamical Gaussian systems},
  author={Barrett, Adam B},
  journal={Physical Review E},
  volume={91},
  number={5},
  pages={052802},
  year={2015},
  publisher={APS}
}

@article{faes2017multiscale,
  title={Multiscale information decomposition: Exact computation for multivariate Gaussian processes},
  author={Faes, Luca and Marinazzo, Daniele and Stramaglia, Sebastiano},
  journal={Entropy},
  volume={19},
  number={8},
  pages={408},
  year={2017},
  publisher={MDPI}
}

@article{kay2018exact,
  title={Exact partial information decompositions for Gaussian systems based on dependency constraints},
  author={Kay, Jim W and Ince, Robin AA},
  journal={Entropy},
  volume={20},
  number={4},
  pages={240},
  year={2018},
  publisher={MDPI}
}

@inproceedings{niu2019measure,
  title={A measure of synergy, redundancy, and unique information using information geometry},
  author={Niu, Xueyan and Quinn, Christopher J},
  booktitle={2019 IEEE International Symposium on Information Theory (ISIT)},
  pages={3127--3131},
  year={2019},
  organization={IEEE}
}

@article{kay2024partial,
  title={A Partial Information Decomposition for Multivariate Gaussian Systems Based on Information Geometry},
  author={Kay, Jim W},
  journal={Entropy},
  volume={26},
  number={7},
  pages={542},
  year={2024},
  publisher={MDPI}
}

@inproceedings{venkatesh2022partial,
  title={Partial information decomposition via deficiency for multivariate gaussians},
  author={Venkatesh, Praveen and Schamberg, Gabriel},
  booktitle={2022 IEEE International Symposium on Information Theory (ISIT)},
  pages={2892--2897},
  year={2022},
  organization={IEEE}
}

@article{venkatesh2023gaussian,
  title={Gaussian partial information decomposition: Bias correction and application to high-dimensional data},
  author={Venkatesh, Praveen and Bennett, Corbett and Gale, Sam and Ramirez, Tamina and Heller, Greggory and Durand, Severine and Olsen, Shawn and Mihalas, Stefan},
  journal={Advances in Neural Information Processing Systems},
  volume={36},
  pages={74602--74635},
  year={2023}
}

@inproceedings{gurushankar2022extracting,
  title={Extracting unique information through markov relations},
  author={Gurushankar, Keerthana and Venkatesh, Praveen and Grover, Pulkit},
  booktitle={2022 58th Annual Allerton Conference on Communication, Control, and Computing (Allerton)},
  pages={1--6},
  year={2022},
  organization={IEEE}
}

@article{pakman2021estimating,
  title={Estimating the unique information of continuous variables},
  author={Pakman, Ari and Nejatbakhsh, Amin and Gilboa, Dar and Makkeh, Abdullah and Mazzucato, Luca and Wibral, Michael and Schneidman, Elad},
  journal={Advances in neural information processing systems},
  volume={34},
  pages={20295--20307},
  year={2021}
}

@article{schick2021partial,
  title={A partial information decomposition for discrete and continuous variables},
  author={Schick-Poland, Kyle and Makkeh, Abdullah and Gutknecht, Aaron J and Wollstadt, Patricia and Sturm, Anja and Wibral, Michael},
  journal={arXiv preprint arXiv:2106.12393},
  year={2021}
}

@article{ehrlich2024partial,
  title={Partial information decomposition for continuous variables based on shared exclusions: Analytical formulation and estimation},
  author={Ehrlich, David A and Schick-Poland, Kyle and Makkeh, Abdullah and Lanfermann, Felix and Wollstadt, Patricia and Wibral, Michael},
  journal={Physical Review E},
  volume={110},
  number={1},
  pages={014115},
  year={2024},
  publisher={APS}
}

@article{bara2025partial,
  title={Partial information decomposition for discrete target and continuous source random variables},
  author={Bar{\`a}, Chiara and Antonacci, Yuri and Iovino, Marta and Lazic, Ivan and Faes, Luca},
  journal={Physical Review E},
  volume={112},
  number={1},
  pages={L012301},
  year={2025},
  publisher={APS}
}

@article{rosas2019quantifying,
  title={Quantifying high-order interdependencies via multivariate extensions of the mutual information},
  author={Rosas, Fernando E and Mediano, Pedro AM and Gastpar, Michael and Jensen, Henrik J},
  journal={Physical Review E},
  volume={100},
  number={3},
  pages={032305},
  year={2019},
  publisher={APS}
}

@article{mediano2025toward,
  title={Toward a unified taxonomy of information dynamics via integrated information decomposition},
  author={Mediano, Pedro AM and Rosas, Fernando E and Luppi, Andrea I and Carhart-Harris, Robin L and Bor, Daniel and Seth, Anil K and Barrett, Adam B},
  journal={Proceedings of the National Academy of Sciences},
  volume={122},
  number={39},
  pages={e2423297122},
  year={2025},
  publisher={National Academy of Sciences}
}

@article{schneidman2003synergy,
  title={Synergy, redundancy, and independence in population codes},
  author={Schneidman, Elad and Bialek, William and Berry, Michael J},
  journal={Journal of Neuroscience},
  volume={23},
  number={37},
  pages={11539--11553},
  year={2003},
  publisher={Society for Neuroscience}
}

@article{luppi2022synergistic,
  title={A synergistic core for human brain evolution and cognition},
  author={Luppi, Andrea I and Mediano, Pedro AM and Rosas, Fernando E and Holland, Negin and Fryer, Tim D and O’Brien, John T and Rowe, James B and Menon, David K and Bor, Daniel and Stamatakis, Emmanuel A},
  journal={Nature neuroscience},
  volume={25},
  number={6},
  pages={771--782},
  year={2022},
  publisher={Nature Publishing Group US New York}
}

@article{gatica2021high,
  title={High-order interdependencies in the aging brain},
  author={Gatica, Marilyn and Cofr{\'e}, Rodrigo and Mediano, Pedro AM and Rosas, Fernando E and Orio, Patricio and Diez, Ibai and Swinnen, Stephan P and Cortes, Jesus M},
  journal={Brain connectivity},
  volume={11},
  number={9},
  pages={734--744},
  year={2021},
  publisher={SAGE Publications Sage CA: Los Angeles, CA}
}

@article{liang2023quantifying,
  title={Quantifying \& modeling multimodal interactions: An information decomposition framework},
  author={Liang, Paul Pu and Cheng, Yun and Fan, Xiang and Ling, Chun Kai and Nie, Suzanne and Chen, Richard and Deng, Zihao and Allen, Nicholas and Auerbach, Randy and Mahmood, Faisal and others},
  journal={Advances in Neural Information Processing Systems},
  volume={36},
  pages={27351--27393},
  year={2023}
}

@inproceedings{dutta2020information,
  title={An information-theoretic quantification of discrimination with exempt features},
  author={Dutta, Sanghamitra and Venkatesh, Praveen and Mardziel, Piotr and Datta, Anupam and Grover, Pulkit},
  booktitle={Proceedings of the AAAI Conference on Artificial Intelligence},
  volume={34},
  number={04},
  pages={3825--3833},
  year={2020}
}

@article{ince2017statistical,
  title={A statistical framework for neuroimaging data analysis based on mutual information estimated via a gaussian copula},
  author={Ince, Robin AA and Giordano, Bruno L and Kayser, Christoph and Rousselet, Guillaume A and Gross, Joachim and Schyns, Philippe G},
  journal={Human brain mapping},
  volume={38},
  number={3},
  pages={1541--1573},
  year={2017},
  publisher={Wiley Online Library}
}

@misc{lyu2026closedformfull,
  title={Closed-Form Gaussian Estimators for Multi-Source Partial Information Decomposition},
  author={Lyu, Aobo and Clark, Andrew and Raviv, Netanel},
  year={2026},
  howpublished={\url{https://wustl.box.com/s/d7u2ukuli7sojtue06vqzaup6ful4ilh}},
  note={Preprint; arXiv version forthcoming},
}

\ifthenelse{\boolean{showAppendix}}
{
\newpage
\appendices
\section{Proof of Theorem~\ref{thm:main_identity}}
  \label{app:main_identity}
 
We prove the claims of Theorem~\ref{thm:main_identity}: 
(i) joint Gaussianity of $(T, \boldY_{\boldA})$,
(ii) the cross-covariance blocks~\eqref{eq:main_T_block}--\eqref{eq:main_cross_block},
(iii) the Schur complement form~\eqref{eq:main_Psi}, and
(iv) positive definiteness of $\Gamma_{\boldA}$ and $\Psi_{\boldA}$.
 
Throughout this proof, $\Pr(\cdot \mid T = t)$ in Definition~\ref{def:conditional_copy_family} is understood as the conditional distribution of the relevant random vector given $T = t$; in the Gaussian setting considered here, this is the standard Gaussian conditional density.
 
  \paragraph{Linear-Gaussian representation.}
By Definition~\ref{def:conditional_copy_family}, each pair $(\boldS_{\cA_a}', T)$ shares the joint Gaussian law of $(\boldS_{\cA_a}, T)$. The conditional distribution of $\boldS_{\cA_a}'$ given $T$ is therefore Gaussian, with mean $\mu_{\cA_a} + B_a (T - \mu_T)$ and covariance $\Delta_{\cA_a}$, where
  \begin{align}
    B_a &\triangleq \Sigma_{\cA_a T}\Sigma_T^{-1},
    \quad
    \Delta_{\cA_a} \triangleq \Sigma_{\cA_a \cA_a} - \Sigma_{\cA_a T}\Sigma_T^{-1}\Sigma_{T \cA_a}.
    \label{eq:cond_mean_app}
  \end{align}
Equivalently, each block admits the linear-Gaussian representation
  \begin{align}
    \boldS_{\cA_a}' &= \mu_{\cA_a} + B_a (T - \mu_T) + \xi_{\cA_a},
    \quad \xi_{\cA_a} \sim \mathcal{N}(0, \Delta_{\cA_a}),
    \label{eq:gen_repr_block}
  \end{align}
where $\xi_{\cA_a}$ is independent of $T$.
By Definition~\ref{def:conditional_copy_family}, distinct blocks are conditionally independent given $T$. Moreover, after subtracting the conditional mean, each residual $\xi_{\cA_a}$ has Gaussian conditional distribution $\mathcal{N}(0, \Delta_{\cA_a})$ given $T = t$, with neither parameter depending on $t$. Hence $\xi_{\cA_a}$ is independent of $T$, and combined with conditional mutual independence given $T$, the residuals $\{\xi_{\cA_a}\}_{a=1}^m$ are mutually independent.
 
  \paragraph{Joint Gaussianity.}
By~\eqref{eq:gen_repr_block}, $\boldY_{\boldA}$ is an affine function of $T$ and the mutually independent Gaussian residuals $\{\xi_{\cA_a}\}$, so $(T, \boldY_{\boldA})$ is jointly Gaussian.
 
  \paragraph{Cross-covariance blocks.}
Equation~\eqref{eq:main_T_block} follows from the shared joint law of $(\boldS_{\cA_a}', T)$ and $(\boldS_{\cA_a}, T)$.
For~\eqref{eq:main_cross_block}: when $a = b$, $\boldS_{\cA_a}'$ shares the marginal of $\boldS_{\cA_a}$, giving $\Sigma_{\cA_a \cA_a}$.
When $a \neq b$, conditional independence yields $\Cov(\boldS_{\cA_a}', \boldS_{\cA_b}' \mid T) = 0$, so by the law of total covariance,
  \begin{align}
    \Cov(\boldS_{\cA_a}', \boldS_{\cA_b}')
    &= \Cov\big(\E[\boldS_{\cA_a}' \mid T], \E[\boldS_{\cA_b}' \mid T]\big) \notag\\
    &= B_a \Sigma_T B_b^\top
    = \Sigma_{\cA_a T}\Sigma_T^{-1}\Sigma_{T \cA_b}.
  \end{align}
 
  \paragraph{Auxiliary covariance decomposition and $\Gamma_{\boldA} \succ 0$.}
Stacking the representations~\eqref{eq:gen_repr_block} across $a$ and using the mutual independence of $\{\xi_{\cA_a}\}$, we obtain
  \begin{align}
    \Gamma_{\boldA} &= D_{\boldA}^{\Delta} + B_{\boldA}\Sigma_T B_{\boldA}^\top,
    \label{eq:Gamma_decomp}
  \end{align}
where $D_{\boldA}^{\Delta} \triangleq \operatorname{diag}(\Delta_{\cA_1}, \ldots, \Delta_{\cA_m})$ collects the within-block conditional covariances and $B_{\boldA}$ stacks the regression matrices $B_a$.
 
The second summand $B_{\boldA}\Sigma_T B_{\boldA}^\top$ in~\eqref{eq:Gamma_decomp} is positive semi-definite, so $\Gamma_{\boldA} \succ 0$ holds whenever $D_{\boldA}^{\Delta} \succ 0$, equivalently $\Delta_{\cA_a} \succ 0$ for each $a$.
The covariance of $(T, \boldS_{\cA_a})$ is a principal submatrix of the positive-definite joint covariance of $(T, S_1, \ldots, S_N)$, hence positive definite, and its Schur complement with respect to $\Sigma_T$ is exactly $\Delta_{\cA_a}$ by~\eqref{eq:cond_mean_app}.
Therefore $\Delta_{\cA_a} \succ 0$ for all $a$, $D_{\boldA}^{\Delta} \succ 0$, and $\Gamma_{\boldA} \succ 0$.
 
  \paragraph{Schur complement form.}
The joint covariance of $(T, \boldY_{\boldA})$ is~\eqref{eq:joint_cov_block}, in which $\Sigma_T \succ 0$ by assumption and $\Gamma_{\boldA} \succ 0$ by the previous paragraph. The Schur complement with respect to $\Gamma_{\boldA}$ yields $\Cov(T \mid \boldY_{\boldA}) = \Sigma_T - \Lambda_{\boldA}^\top \Gamma_{\boldA}^{-1} \Lambda_{\boldA}$, proving~\eqref{eq:main_Psi}.
 
  \paragraph{$\Psi_{\boldA} \succ 0$.}
Intuitively, the map $(T, \xi_{\cA_1}, \ldots, \xi_{\cA_m}) \mapsto (T, \boldY_{\boldA})$ is invertible block-triangular, so the joint covariance of $(T, \boldY_{\boldA})$ is positive definite, hence so is its Schur complement $\Psi_{\boldA}$. Explicitly, substituting~\eqref{eq:Gamma_decomp} into~\eqref{eq:main_Psi} and using $\Lambda_{\boldA} = B_{\boldA}\Sigma_T$, the Woodbury matrix identity gives
  \begin{align}
    \Psi_{\boldA}
    &= \Sigma_T - \Sigma_T B_{\boldA}^\top \big(D_{\boldA}^{\Delta} + B_{\boldA}\Sigma_T B_{\boldA}^\top\big)^{-1} B_{\boldA}\Sigma_T \notag\\
    &= \big(\Sigma_T^{-1} + B_{\boldA}^\top (D_{\boldA}^{\Delta})^{-1} B_{\boldA}\big)^{-1}.
  \end{align}
Since $\Sigma_T^{-1} \succ 0$ and $B_{\boldA}^\top (D_{\boldA}^{\Delta})^{-1} B_{\boldA} \succeq 0$, the inverse on the right is positive definite, so $\Psi_{\boldA} \succ 0$.
 
\section{Derivations of Corollary~\ref{cor:closed_forms}}
  \label{app:corollaries}
 
For jointly Gaussian $(T, \boldY_{\boldA})$ with $\Psi_{\boldA} \succ 0$,
  \begin{align}
    h(T \mid \boldY_{\boldA})
    &= \tfrac{1}{2}\log\big((2\pi e)^{d_T} \det \Psi_{\boldA}\big);
    \label{eq:app_cond_entropy}
  \end{align}
the constant $\tfrac{1}{2}\log\big((2\pi e)^{d_T}\big)$ cancels in every entropy difference below.
 
  \paragraph{Two-source redundancy \eqref{eq:cor_red}.}
By Definition~\ref{def:measures_quantities}, $\operatorname{Red}(S_1, S_2 \to T) = I(S_1'; S_2')$.
Under joint Gaussianity, $I(S_1'; S_2') = \tfrac{1}{2}\log\frac{\det \Sigma_{11} \det \Sigma_{22}}{\det \operatorname{Cov}((S_1', S_2'))}$, where $\operatorname{Cov}((S_1', S_2'))$ denotes the joint covariance of the stacked vector. Theorem~\ref{thm:main_identity} with $\boldA = \boldU_1 = \{\{1\}, \{2\}\}$ identifies this joint covariance as $\Gamma_{\boldU_1}$, yielding \eqref{eq:cor_red}.
 
  \paragraph{General unique information \eqref{eq:cor_unique}.}
By Definition~\ref{def:measures_quantities}, $\operatorname{Un}(S_i \to T \mid \boldS_{\setminus i}) = h(T \mid \boldS'_{\setminus i}) - h(T \mid S_i', \boldS'_{\setminus i})$.
The first term has conditional covariance $\Psi_{\boldV_i}$ via the single-block family case of Theorem~\ref{thm:main_identity}; the second has $\Psi_{\boldU_i}$ via Theorem~\ref{thm:main_identity}.
Substituting into \eqref{eq:app_cond_entropy} yields \eqref{eq:cor_unique}.
 
  \paragraph{Total synergistic effect \eqref{eq:cor_tse}.}
Summing \eqref{eq:cor_SE_K} over $K = 2, \ldots, N$ telescopes:
  \begin{align}
    \sum_{K=2}^{N}
    \tfrac{1}{2}\log\frac{\det \Psi_{\boldC_{K-1}}}
    {\det \Psi_{\boldC_K}}
    &= \tfrac{1}{2}\log\frac{\det \Psi_{\boldC_1}}{\det \Psi_{\boldC_N}}.
  \end{align}
 
  \paragraph{$K$-th order synergistic effect \eqref{eq:cor_SE_K}.}
By Definition~\ref{def:measures_quantities}, $\operatorname{SE}_K = h(T \mid \boldY_{\boldC_{K-1}}) - h(T \mid \boldY_{\boldC_K})$.
Substituting \eqref{eq:app_cond_entropy} for both terms and applying Theorem~\ref{thm:main_identity} to compute each $\Psi_{\boldC_K}$ gives \eqref{eq:cor_SE_K}.
 
  \paragraph{Narrow synergy.}
Specializing \eqref{eq:cor_SE_K} at $K = N$ gives $\operatorname{Syn} = \operatorname{SE}_N = \tfrac{1}{2}\log\frac{\det \Psi_{\boldC_{N-1}}}{\det \Psi_{\boldC_N}}$, where $\boldC_N = \{[N]\}$ is a single-block family with $\Psi_{\boldC_N} = \Cov(T \mid \boldS_{[N]})$.
 
  \section{Proof of Theorem~\ref{thm:structural}}
  \label{app:structural}
 
  \paragraph{(i) Blockwise affine invariance.}
Let $\widetilde{T} = A_T T + b_T$ and $\widetilde{S}_i = A_i S_i + b_i$ with each $A_T, A_i$ invertible.
Shifts do not affect covariances, so we may take $b_T = 0$ and $b_i = 0$.
 
For any subset $\cA \subseteq [N]$, write $A_{\cA} = \operatorname{diag}((A_i)_{i \in \cA})$.
Then $\widetilde{\Sigma}_T = A_T \Sigma_T A_T^\top$, $\widetilde{\Sigma}_{\cA \cA} = A_{\cA} \Sigma_{\cA \cA} A_{\cA}^\top$, and $\widetilde{\Sigma}_{\cA T} = A_{\cA} \Sigma_{\cA T} A_T^\top$.
Substituting into the conditional-independent construction, the auxiliary vector for the transformed system has the same distribution as $A_{\cA_a} \boldS_{\cA_a}'$ blockwise, so $\widetilde{\Lambda}_{\boldA} = (\bigoplus_a A_{\cA_a}) \Lambda_{\boldA} A_T^\top$ and $\widetilde{\Gamma}_{\boldA} = (\bigoplus_a A_{\cA_a}) \Gamma_{\boldA} (\bigoplus_a A_{\cA_a})^\top$, where $\bigoplus_a A_{\cA_a} \triangleq \operatorname{diag}(A_{\cA_1}, \ldots, A_{\cA_m})$ denotes the block-diagonal matrix stacking the source-side transforms.
A direct computation using the Schur complement form \eqref{eq:main_Psi} yields
  \begin{align}
    \widetilde{\Psi}_{\boldA} &= A_T \Psi_{\boldA} A_T^\top.
    \label{eq:Psi_affine_transform}
  \end{align}
Note that the source-side determinants $\det(A_i)$ do not appear in this transformation: the source-side factors $\bigoplus_a A_{\cA_a}$ enter $\widetilde{\Lambda}_{\boldA}$ and $\widetilde{\Gamma}_{\boldA}$ in matched positions and cancel in the Schur complement structure of $\widetilde{\Psi}_{\boldA}$.
For any two families $\boldA_1, \boldA_2$,
  \begin{align}
    &\tfrac{1}{2}\log
    \det\big(A_T \Psi_{\boldA_1} A_T^\top\big)
    - \tfrac{1}{2}\log
    \det\big(A_T \Psi_{\boldA_2} A_T^\top\big)
    \notag\\
    &\quad = \tfrac{1}{2}\log
    \det \Psi_{\boldA_1}
    - \tfrac{1}{2}\log
    \det \Psi_{\boldA_2},
  \end{align}
since $\log\det(A_T A_T^\top)$ cancels.
Each $\theta$ in Section~\ref{sec:gaussian_estimators} that is a difference of $\log\det\Psi_{\boldA}$ terms (unique information, $\operatorname{SE}_K$, narrow synergy, TSE) therefore satisfies $\theta(\widetilde{T}, \widetilde{\boldS}_{[N]}) = \theta(T, \boldS_{[N]})$.
 
For two-source redundancy ($N = 2$), $\widetilde{\Sigma}_{11} = A_1 \Sigma_{11} A_1^\top$ and $\widetilde{\Sigma}_{22} = A_2 \Sigma_{22} A_2^\top$. The off-diagonal block of $\Gamma_{\boldU_1}$ transforms as $\widetilde{\Sigma}_{1T}\widetilde{\Sigma}_T^{-1}\widetilde{\Sigma}_{T2} = A_1 \Sigma_{1T} \Sigma_T^{-1} \Sigma_{T2} A_2^\top$ since $A_T$ cancels, giving $\widetilde{\Gamma}_{\boldU_1} = \operatorname{diag}(A_1, A_2)\Gamma_{\boldU_1}\operatorname{diag}(A_1, A_2)^\top$. The numerator $\det\widetilde{\Sigma}_{11}\det\widetilde{\Sigma}_{22}$ and the denominator $\det\widetilde{\Gamma}_{\boldU_1}$ in~\eqref{eq:cor_red} both pick up $\det(A_1)^2\det(A_2)^2$, which cancels in the ratio.
 
  \paragraph{(ii) Source-permutation symmetry/equivariance.}
Let $\pi$ be a permutation of $[N]$ and set $\widetilde{S}_i = S_{\pi(i)}$.
Under this relabeling, any subset family $\widetilde{\boldA}$ in the relabeled system corresponds to the original family $\pi(\widetilde{\boldA}) = \{\pi(\cA) : \cA \in \widetilde{\boldA}\}$ in the original sources, and $\widetilde{\Psi}_{\widetilde{\boldA}} = \Psi_{\pi(\widetilde{\boldA})}$.
 
Since $\boldC_K = \binom{[N]}{K}$ is the family of all $K$-subsets, $\pi(\boldC_K) = \boldC_K$, so $\widetilde{\Psi}_{\boldC_K} = \Psi_{\boldC_K}$, and $\operatorname{SE}_K$, $\operatorname{Syn}$, $\operatorname{TSE}$ are invariant.
 
For unique information, the relabeled family $\widetilde{\boldU}_i = \{\{i\}, [N] \setminus \{i\}\}$ corresponds in the original sources to $\pi(\widetilde{\boldU}_i) = \{\{\pi(i)\}, [N] \setminus \{\pi(i)\}\} = \boldU_{\pi(i)}$.
Hence $\operatorname{Un}(\widetilde{S}_i \to T \mid \widetilde{\boldS}_{\setminus i}) = \operatorname{Un}(S_{\pi(i)} \to T \mid \boldS_{\setminus \pi(i)})$.
 
Two-source redundancy is symmetric in $(S_1, S_2)$ directly from~\eqref{eq:cor_red}: swapping $S_1 \leftrightarrow S_2$ exchanges $\det\Sigma_{11} \leftrightarrow \det\Sigma_{22}$ in the numerator (whose product is symmetric) and applies a conformal permutation to $\Gamma_{\boldU_1}$ that leaves $\det\Gamma_{\boldU_1}$ invariant.
 
  \paragraph{(iii) Additivity over independent systems.}
Independence of $(T^{(a)}, \boldS_{[N]}^{(a)})$ and $(T^{(b)}, \boldS_{[N]}^{(b)})$ implies that all cross-system covariance blocks vanish.
Hence, after a conformal permutation grouping all system-$(a)$ coordinates before all system-$(b)$ coordinates, the joint covariance of $(T, \boldS_{[N]})$ with $T = (T^{(a)}, T^{(b)})$ and $S_i = (S_i^{(a)}, S_i^{(b)})$ is block diagonal:
  \begin{align}
    \Sigma_T &=
    \begin{bmatrix}
      \Sigma_T^{(a)} & 0 \\
      0 & \Sigma_T^{(b)}
    \end{bmatrix},
    \quad
    \Sigma_{\cA \cB} =
    \begin{bmatrix}
      \Sigma_{\cA \cB}^{(a)} & 0 \\
      0 & \Sigma_{\cA \cB}^{(b)}
    \end{bmatrix},
  \end{align}
and similarly for $\Sigma_{\cA T}$.
Substituting these block diagonals into the main identity \eqref{eq:main_Psi} of Theorem~\ref{thm:main_identity} gives
  \begin{align}
    \Psi_{\boldA} &=
    \begin{bmatrix}
      \Psi_{\boldA}^{(a)} & 0 \\
      0 & \Psi_{\boldA}^{(b)}
    \end{bmatrix},
  \end{align}
so $\det \Psi_{\boldA} = \det \Psi_{\boldA}^{(a)} \cdot \det \Psi_{\boldA}^{(b)}$ and $\log\det \Psi_{\boldA} = \log\det \Psi_{\boldA}^{(a)} + \log\det \Psi_{\boldA}^{(b)}$.
Each $\theta$ expressible as a difference of $\log\det\Psi_{\boldA}$ terms therefore splits additively across the two systems. For two-source redundancy, the same block-diagonal argument applies to $\Sigma_{11}$, $\Sigma_{22}$, and $\Gamma_{\boldU_1}$, and the log-determinant ratio~\eqref{eq:cor_red} splits additively.
 
  \paragraph{(iv) Plug-in consistency.}
Each Gaussian PID measure $\theta(\Sigma)$ is a composition of continuous operations on blocks of $\Sigma$ (selection of principal sub-blocks, matrix multiplication, matrix inversion, Schur complementation, determinant, and logarithm), each continuous on the open positive-definite cone.
Under $\Sigma \succ 0$, Theorem~\ref{thm:main_identity} ensures $\Gamma_{\boldA}, \Psi_{\boldA} \succ 0$ for every family $\boldA$ involved, so $\theta$ is continuous at $\Sigma$.
Since the positive-definite cone is open, $\widehat{\Sigma}_M \xrightarrow{\mathrm{a.s.}} \Sigma$ implies that, almost surely, $\widehat{\Sigma}_M$ eventually lies in a neighborhood of $\Sigma$ on which all required Schur complements remain positive definite, so $\theta(\widehat{\Sigma}_M)$ is well defined for all sufficiently large $M$, and the continuous mapping theorem yields $\theta(\widehat{\Sigma}_M) \xrightarrow{\mathrm{a.s.}} \theta(\Sigma)$.
Two-source redundancy is likewise a composition of block selection, multiplication, inversion (via $\Gamma_{\boldU_1}$), determinant, and logarithm on the positive-definite cone, hence continuous under $\Sigma \succ 0$; the same argument yields plug-in consistency.
 
  \section{Proof of Proposition~\ref{prop:ridge}}
  \label{app:ridge}
 
For $\widehat{\Sigma} \succeq 0$ and $\lambda > 0$, $\widehat{\Sigma}_\lambda = \widehat{\Sigma} + \lambda I \succeq \lambda I \succ 0$.
Every principal submatrix of $\widehat{\Sigma}_\lambda$ is therefore positive definite, and so are all Schur complements built from these submatrices.
By Theorem~\ref{thm:main_identity}, $\Gamma_{\boldA}$ and $\Psi_{\boldA}$ constructed from $\widehat{\Sigma}_\lambda$ are positive definite, so $\theta(\widehat{\Sigma}_\lambda)$ is finite.
The mapping $\widehat{\Sigma} \mapsto \theta(\widehat{\Sigma}_\lambda)$ is a composition of continuous matrix operations on the open positive-definite cone, hence continuous in $\widehat{\Sigma}$, proving (a).
 
For (b), if $\widehat{\Sigma} \succ 0$ and all required Schur complements of $\widehat{\Sigma}$ are positive definite, then $\theta$ is continuous at $\widehat{\Sigma}$ by the same argument as in Appendix~\ref{app:structural}~(iv).
Since $\widehat{\Sigma}_\lambda \to \widehat{\Sigma}$ as $\lambda \to 0$, $\theta(\widehat{\Sigma}_\lambda) \to \theta(\widehat{\Sigma})$ by continuity.
The population analogue, with $\Sigma$ in place of $\widehat{\Sigma}$, follows by the identical argument. 
 
  \paragraph{Remark on invariance.}
The isotropic ridge $\widehat{\Sigma} + \lambda I$ depends on the coordinate scale of $(T, \boldS_{[N]})$, so the blockwise affine invariance of Theorem~\ref{thm:structural}(i) does not hold for fixed $\lambda > 0$. It is recovered in the limit $\lambda \downarrow 0$ whenever the unregularized estimator is well defined.
 
  \section{Empirical Robustness and Ridge Behavior}
  \label{app:ridge_exp}
 
\textit{Empirical observation.}
In this benchmark, plug-in estimation is numerically stable in the small-sample regime ($M \approx d$), with ridge regularization providing improvement only at the smallest $M$ and becoming unnecessary as $M$ grows.
 
\paragraph{Setup}
We probe near-singular regimes on the Section~\ref{subsec:exp_recovery} benchmark, where the joint dimension is $d = 7$.
We vary the sample size $M \in \{10, 12, 15, 25, 50, 100, 500\}$ and the ridge parameter $\lambda \in \{0, 10^{-8}, 10^{-6}, 10^{-4}, 10^{-2}, 10^{-1}, 1\}$, repeating each $(M, \lambda)$ configuration over $50$ independent trials.
The estimand is $\operatorname{TSE}$, with population value computed in closed form from the construction.
 
\paragraph{Results}
The closed-form Gaussian estimator is numerically robust. For $\lambda > 0$, Cholesky factorization of $\Gamma$ succeeds in all trials, as guaranteed by Proposition~\ref{prop:ridge}. The unregularized case ($\lambda = 0$) succeeds in all $7 \times 50 = 350$ trials across the seven sample sizes, including all $50$ trials at $M = 10$ where $M / d \approx 1.4$, consistent with the standard Wishart result that $\widehat{\Sigma}_M \succ 0$ holds almost surely once $M \geq d + 1$.
Figure~\ref{fig:exp3} maps the relative error across the $M \times \lambda$ grid (truncated to $\lambda \leq 10^{-4}$; in this benchmark, larger $\lambda$ substantially shrinks the estimated TSE and is omitted for readability).
The optimal $\lambda^{*}$ shifts from $10^{-4}$ at $M = 10$ to $0$ at $M \geq 50$: a small ridge reduces relative error when $M$ approaches $d$, while the unregularized estimator becomes preferable as $M$ grows. Proposition~\ref{prop:ridge} guarantees finiteness and continuity of the regularized estimator; the bias-reduction effect is benchmark-specific.
 
\begin{figure}[!htbp]
  \centering
  \includegraphics[width=\columnwidth]{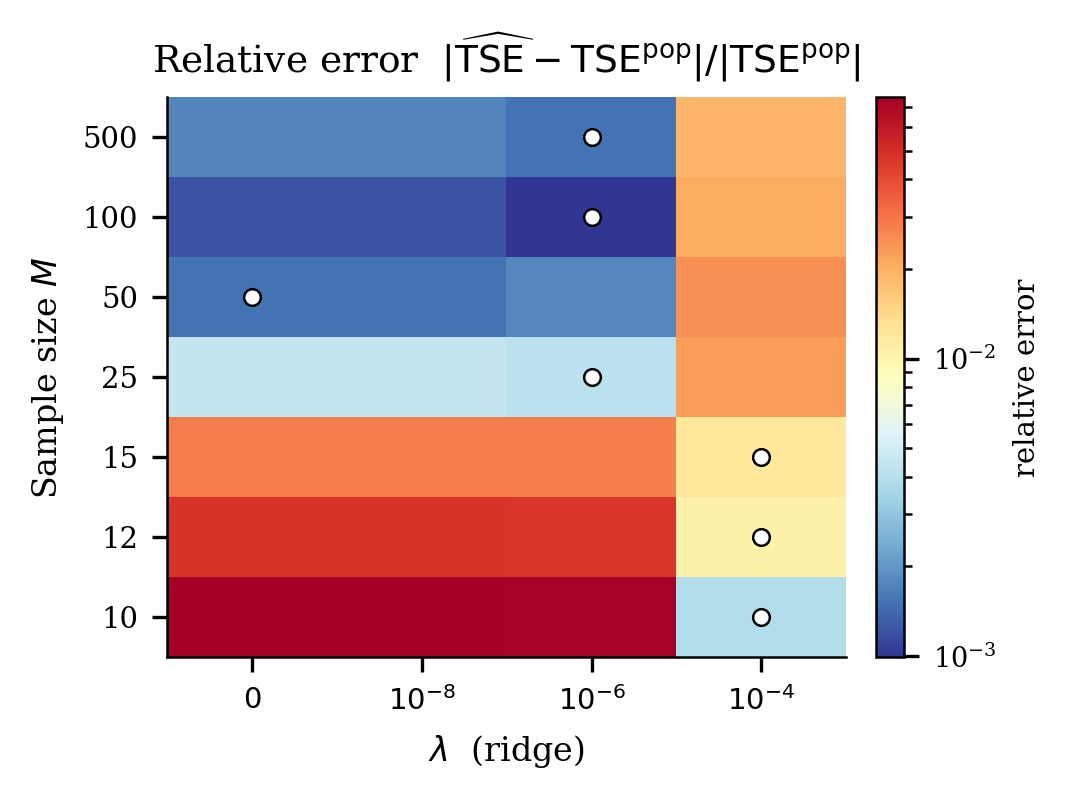}
  \caption{Relative error
  $|\widehat{\operatorname{TSE}}_\lambda
  - \operatorname{TSE}|
  / |\operatorname{TSE}|$ across the $M \times
  \lambda$ grid (mean over $50$ trials). White circles mark
  the optimal $\lambda^{*}$ for each $M$.}
  \label{fig:exp3}
\end{figure}
 
\section{Finite-Sample Convergence}
  \label{app:convergence}
 
\textit{Empirical observation.}
The plug-in Gaussian PID estimators converge to their population values as the sample size $M$ grows, with bias decaying rapidly and Monte-Carlo standard deviation visually consistent with the parametric $M^{-1/2}$ rate.
 
\paragraph{Setup}
We reuse the Section~\ref{subsec:exp_recovery} 5-source benchmark exactly and vary $M \in \{50, 100, 200, 500, 1000, 2000, 5000, 10000\}$ over $100$ independent trials per $M$.
We track five plug-in estimators with population values (in nats, computed in closed form from the construction)
\begin{align*}
&\operatorname{SE}_2 = 3.727,\quad \operatorname{SE}_3 = 3.268,\quad \operatorname{TSE} = 6.443, \\
&\operatorname{Syn}(\{S_1, S_2\} \to T) = 3.140,\\
&\operatorname{Syn}(\{S_3, S_4, S_5\} \to T) = 2.853.
\end{align*}
The smaller-magnitude signed components $\operatorname{SE}_4 \approx +0.34$ and $\operatorname{SE}_5 \approx -0.90$ are omitted from this plot for visual clarity; they show the same qualitative convergence pattern and are excluded only to keep the vertical scale of Figure~\ref{fig:exp4} compact.
 
\paragraph{Results}
All five plug-in estimators converge to their population values across the full $M$ range (Figure~\ref{fig:exp4}).
$\operatorname{SE}_2$, $\operatorname{SE}_3$, $\operatorname{TSE}$, $\operatorname{Syn}(\{S_1,S_2\} \to T)$, and $\operatorname{Syn}(\{S_3,S_4,S_5\} \to T)$ exhibit negligible bias for $M \geq 500$ and Monte-Carlo standard deviation visually consistent with $M^{-1/2}$, reaching relative SDs below $0.3\%$ at $M = 10^4$, with $\operatorname{TSE}$ converging fastest in relative terms ($0.18\%$).
No quantity exhibits divergent or systematically biased behavior, illustrating the consistency of Theorem~\ref{thm:structural}(iv), with empirically vanishing bias and numerical conditioning preserved across the full benchmark regime.
 
\begin{figure}[!htbp]
  \centering
  \includegraphics[width=\columnwidth]{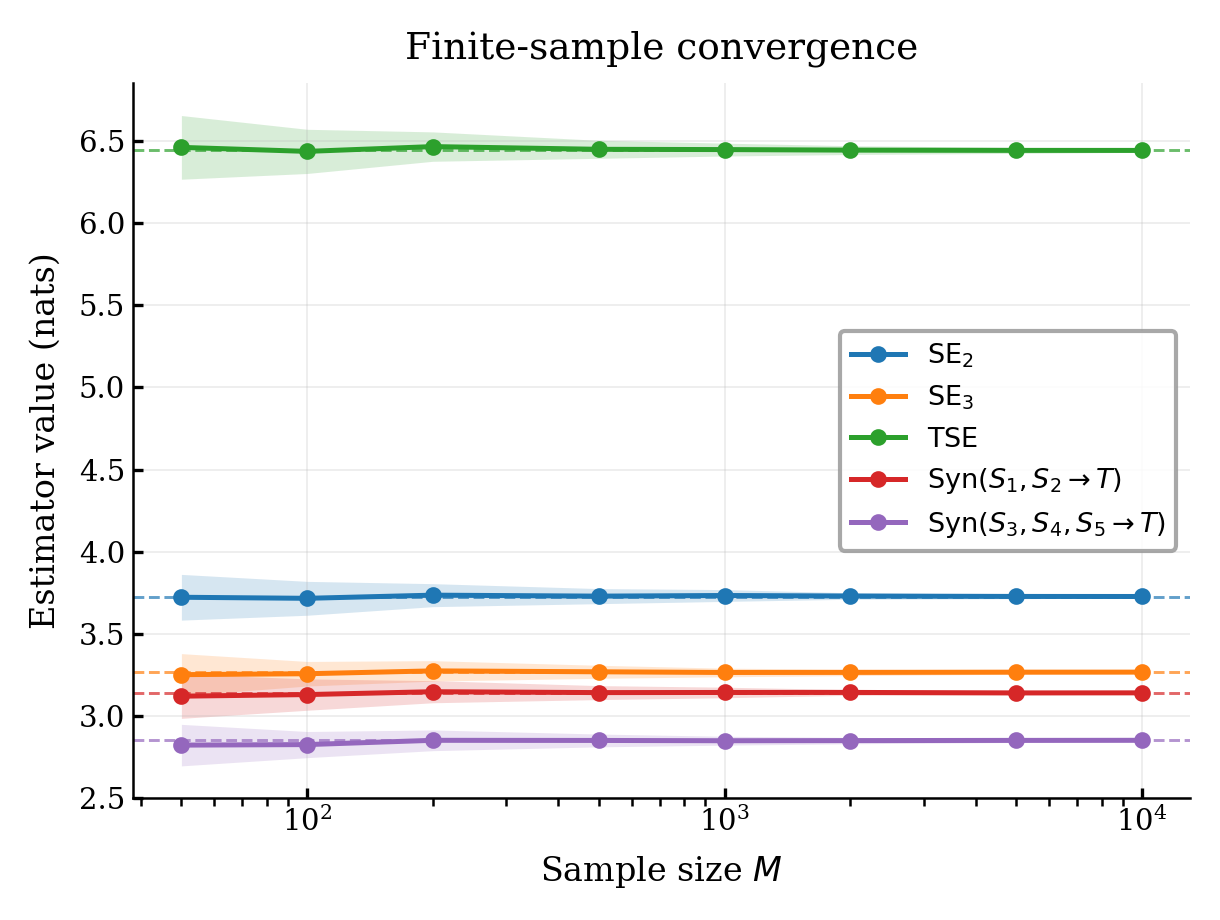}
  \caption{Finite-sample convergence on the Section~\ref{subsec:exp_recovery} benchmark (mean $\pm$ std over $100$ trials). Estimators $\operatorname{SE}_2$, $\operatorname{SE}_3$, $\operatorname{TSE}$, and the subset narrow synergies $\operatorname{Syn}(\{S_1, S_2\} \to T)$ and $\operatorname{Syn}(\{S_3, S_4, S_5\} \to T)$ versus sample size $M$ (log scale); population values shown as dashed horizontal lines.}
  \label{fig:exp4}
\end{figure}
 
\section{Two-Source Estimator Comparison}
  \label{app:n2_comparison}
 
\textit{Empirical observation.}
For two scalar Gaussian sources, our estimator agrees with existing two-source PID definitions in trivial limits and assigns non-zero unique information to both sources whenever both contribute, distinguishing it from MMI-derived measures.
 
\paragraph{Setup}
We compare four estimators on five jointly Gaussian scalar configurations spanning the redundancy/unique/synergy triangle:
(i) \emph{Pure redundancy} $S_i = T + \epsilon_i$ for $i = 1, 2$ with i.i.d.\ $\epsilon_i$;
(ii) \emph{Pure unique} $T = S_1 + \epsilon_T$ with $S_2 \perp T$;
(iii) \emph{Pure synergy} $T = S_1 + S_2 + \epsilon_T$ with $S_1 \perp S_2$;
(iv) \emph{Mixed-correlated}: variant of (iii) with $\operatorname{Corr}(S_1, S_2) = 0.3$;
(v) \emph{Mixed-asymmetric} $T = 2 S_1 + S_2 + \epsilon$ with $S_1 \perp S_2$.
All noise terms ($\epsilon_i$ and $\epsilon_T$) are i.i.d.\ $\mathcal{N}(0, 1)$ and mutually independent. In Pure redundancy, $T \sim \mathcal{N}(0, 1)$ is sampled first and $S_i$ are derived; in the remaining four configurations, $S_1$ and $S_2$ are sampled first with $\operatorname{Var}(S_i) = 1$ (jointly Gaussian with $\operatorname{Cov}(S_1, S_2) = 0.3$ in Mixed-correlated), and $T$ is then constructed by the displayed formula. All variables are centered before covariance estimation.
We compare against the Barrett MMI~\cite{barrett2015exploration}, Venkatesh--Schamberg $\delta$-PID~\cite{venkatesh2022partial}, and Venkatesh et al.\ $\widetilde G$-PID~\cite{venkatesh2023gaussian}, computed via the authors' released \texttt{gpid} package (\url{https://github.com/praveenv253/gpid}).
Each configuration uses $M = 1000$ samples averaged over $50$ trials. Sample standard deviations across these trials are at most $0.03$ for every cell of Table~\ref{tab:n2_comparison}, so we tabulate sample means only.
Kay--Ince $I_{\mathrm{dep}}$~\cite{kay2018exact} and Kay $I_{\mathrm{ig}}$~\cite{kay2024partial} are released only as R code (\url{https://github.com/JWKay/PID}); their inclusion is left to future work.
 
\paragraph{Findings}
Two observations emerge from Table~\ref{tab:n2_comparison}.
 
\noindent (i) \emph{Three Venkatesh-style measures coincide for scalar Gaussian inputs.}
Across all five configurations, Barrett's MMI, $\delta$-PID, and $\widetilde G$-PID return numerically identical decompositions to the precision of our experiment.
This is consistent with the scalar-Gaussian $\delta$-PID $=$ MMI equivalence of~\cite[Theorem 1]{venkatesh2022partial} and indicates that $\widetilde G$-PID coincides with both on these examples; we do not claim a general equivalence theorem.
 
\noindent (ii) \emph{Our estimator assigns dual non-zero unique information.}
In four of five configurations (Pure redundancy, Pure synergy, Mixed-correlated, Mixed-asymmetric), each source carries an independent component about $T$, and our estimator is the only one in the comparison that assigns positive $\operatorname{Un}_1$ \emph{and} $\operatorname{Un}_2$ simultaneously.
The MMI/$\delta$-PID/$\widetilde G$-PID family assigns positive unique information only to the source with the larger marginal $I(S_i; T)$ (and zero to the other); in symmetric configurations where $I(S_1; T) = I(S_2; T)$, both unique components vanish. This conflates uniqueness with the ranking of marginal mutual informations.
 
\begin{table}[t]
  \centering
  \caption{Two-source Gaussian PID on five controlled configurations.}
  \label{tab:n2_comparison}
  \footnotesize
  \renewcommand{\arraystretch}{1.15}
  \setlength{\tabcolsep}{3pt}
  \begin{tabular}{@{}llcccc@{}}
    \hline
    Configuration & Estimator & $\operatorname{Red}$ & $\operatorname{Un}_1$ & $\operatorname{Un}_2$ & $\operatorname{Syn}$ \\
    \hline
    Pure redundancy
      & Ours
      & $0.14$ & $0.20$ & $0.21$ & $0.00$ \\
      & MMI/$\delta$/$\widetilde G$
      & $0.33$ & $0.01$ & $0.01$ & $0.19$ \\
    \hline
    Pure unique ($S_2 \perp T$)
      & All four
      & $0.00$ & $0.35$ & $0.00$ & $0.00$ \\
    \hline
    Pure synergy
      & Ours
      & $0.06$ & $0.15$ & $0.14$ & $0.20$ \\
      & MMI/$\delta$/$\widetilde G$
      & $0.19$ & $0.01$ & $0.01$ & $0.34$ \\
    \hline
    Mixed-correlated
      & Ours
      & $0.13$ & $0.19$ & $0.19$ & $0.13$ \\
      & MMI/$\delta$/$\widetilde G$
      & $0.31$ & $0.01$ & $0.01$ & $0.31$ \\
    \hline
    Mixed-asymmetric
      & Ours
      & $0.06$ & $0.50$ & $0.03$ & $0.32$ \\
      & MMI/$\delta$/$\widetilde G$
      & $0.09$ & $0.46$ & $0.00$ & $0.35$ \\
    \hline
  \end{tabular}
 
  \vspace{2pt}
  {\scriptsize Plug-in mean over $50$ trials at $M = 1000$. \emph{Ours}: this paper's closed-form conditional-copy decomposition. \emph{MMI/$\delta$/$\widetilde G$}: Barrett MMI~\cite{barrett2015exploration}, Venkatesh--Schamberg $\delta$-PID~\cite{venkatesh2022partial}, and $\widetilde G$-PID~\cite{venkatesh2023gaussian}, computed via the \texttt{gpid} package, which coincide numerically and are collapsed into a single row.\par}
\end{table}
}{}
\end{document}